\documentclass[11pt, a4paper]{article} %article

\usepackage[utf8]{inputenc}
\usepackage[T1]{fontenc}

\usepackage{amsmath,amssymb,amsthm}
\usepackage[american]{babel}
\usepackage{microtype,ellipsis}
\usepackage{verbatim}
\usepackage{tabularx}
\usepackage{float}
\usepackage{amsfonts}
\usepackage{mathtools}   % loads »amsmath«
\usepackage{algorithm}
\usepackage[pagebackref]{hyperref}
\usepackage{myprobdescr}
\usepackage{tikz}
\usepackage[noend]{algpseudocode}
\usepackage{caption}
\usepackage{subcaption}
\usepackage{cleveref}
\usepackage{paralist}
\usepackage{a4wide}

\usepackage{authblk}

\usepackage{scrpage2}

\usepackage[textsize=footnotesize]{todonotes}

\usepackage[numbers,sort]{natbib}

\newtheorem{theorem}{Theorem}
\newtheorem{corollary}{Corollary}

\newtheorem{conjecture}{Conjecture}

\theoremstyle{definition}
\newtheorem{definition}{Definition}

{\bfseries}{\normalfont}

\AtBeginDocument{%
    \renewcommand{\ref}[1]{\mbox{\autoref{#1}}}

}

\usetikzlibrary{arrows,decorations.pathmorphing,backgrounds,positioning,fit,patterns}
\usetikzlibrary{shapes}
\usetikzlibrary{arrows,calc}
\usetikzlibrary{positioning,backgrounds,patterns,calc}

\newcommand{\problemdef}[3]{
  \begin{center}
    \begin{minipage}{0.95\textwidth}
      \noindent
      \textsc{#1}

      \vspace{2pt}
      \setlength{\tabcolsep}{3pt}
      \begin{tabularx}{\textwidth}{@{}lX@{}}
        \textbf{Input:} 		& #2 \\
        \textbf{Question:} 	& #3
      \end{tabularx}
    \end{minipage}
  \end{center}
}

\newcommand{\PVG}{\text{PVG}}
\newcommand{\PVGS}{\text{PVGs}}

\begin{document}

\title{Computational Complexity Aspects of Point Visibility Graphs}

\author[1]{Anne-Sophie Himmel}

\author[1]{Clemens Hoffmann}

\author[1]{Pascal Kunz}

\author[1]{Vincent Froese}

\author[1,2]{Manuel~Sorge\footnote{Supported by the People Programme (Marie Curie Actions) of the European Union's Seventh Framework Programme (FP7/2007-2013) under REA grant agreement number 631163.11 and by the Israel Science Foundation (grant no. 551145/14).}}

\affil[1]{Institut für Softwaretechnik und Theoretische Informatik, TU Berlin, Berlin,\newline \texttt{vincent.froese@tu-berlin.de}}

\affil[2]{Department of Industrial Engineering and Management, Ben Gurion University of the Negev, Beer~Sheva, Israel, \texttt{sorge@post.bgu.ac.il}}

\date{}

\maketitle

\begin{abstract}
  A point visibility graph is a graph induced by a set of points in the plane where the vertices of the graph represent the points in the point set and two vertices are adjacent if and only if no other point from the point set lies on the line segment between the two corresponding points. The set of all point visibility graphs form a graph class which is examined from a computational complexity perspective in this paper. We show NP-hardness for several classic graph problems on point visibility graphs such as \textsc{Feedback Vertex Set}, \textsc{Longest Induced Path}, \textsc{Bisection} and \textsc{$\mathcal{F}$-free Vertex Deletion} (for certain sets~$\mathcal{F}$).
  % We show that for the \textsc{Dominating Set} problem there exists an upper bound on the solution size depending on the minimal degree of the graph. 
Furthermore, we consider the complexity of the \textsc{Dominating Set} problem on point visibility graphs of points on a grid.
\end{abstract}

\section{Introduction}
Visibility graphs are a way to encode the information that certain objects are visible from one another or not. The objects correspond to the vertices and there is an edge between two vertices if and only if the two corresponding objects are visible from each other, for some specified definition of visibility. Different kinds of visibility graphs have been studied, like rectangle visibility graphs~\cite{Pet16}, hypercube visibility graphs~\cite{PW17}, 
segment visibility graphs~\cite{GG13}, and polygon visibility graphs (bearing a slightly different meaning). They find their application in many real world problems, for example, in computing Euclidean shortest paths in the presence of obstacles in the field of robotics \cite{deBerg2008}, decomposition of two-dimensional shapes by graph theoretic clustering in the field of object recognition \cite{Shapiro1979}, or even in the diagnosis of Alzheimer’s disease \cite{Ahmadlou2010}. In point visibility graphs (PVGs) our objects are simply points in the plane and two points are \emph{visible} if there is a direct line between them, that is, a line that does not contain any other point. PVGs can be thought of the extreme case of other visibility graphs when our viewpoint on the objects is far away. In this case the shapes and diameter of the objects become negligible as the objects shrink to points. The visibility relation between points of a point set in the plane (represented by PVGs) is thus a fundamental structure in computational geometry~\cite{deBerg2008}.

In this work, we adopt an algorithmic perspective on the class of point visibility graphs.
In doing so, we intend to bring this practically relevant graph class to the attention of a broader audience in order to motivate research focused on solving computational graph problems for this graph class.
We start by studying several classical graph problems that are NP-complete in general and investigate whether they become polynomial-time solvable on PVGs. It turns out that many of the problems remain NP-complete on point visibility graphs.

\subsection{Preliminaries and Properties of Point Visibility Graphs}
All graphs in this paper are undirected, without self-loops or parallel edges. We use standard graph notation (see, e.g. \citet{Die16}). We start with the definition of point visibility graphs. 
\begin{definition}
  A graph $G=(V,E)$ with $V=\{v_1,\ldots, v_n\}$ is a \emph{point visibility graph} (\PVG), if there exists a set of points $P=\{p_1,\ldots,p_n\}$ in the plane (each point $p_i$ corresponds to vertex $v_i$) such that $\{v_i,v_j\}\in E$ if and only if there exists no other point in~$P$ on the line segment between $p_i$ and $p_j$, that is, $p_i$ and $p_j$ are \emph{visible} to each other.
  The point set~$P$ is also called the \emph{visibility embedding} of~$G$.
\end{definition}

In short, a \PVG\ is a graph that has a visibility embedding in the plane.
\autoref{figure:PVG} shows a \PVG\ denoted $G$ and its visibility embedding in the plane.

We can divide the graph class of \PVGS\ into two subclasses: Paths and non-path graphs.
Every path is a \PVG~and it holds that the diameter is equal to the length of the path.
On the other hand, every non-path \PVG\ has diameter two~\cite{kara2005chromatic}.
For non-path \PVGS, this is due to the fact that for each pair of points in the plane that are not visible to each other there exists a third point that can see both: The point closest to the line segment between the two points.
Another property of interest of non-path \PVGS\ is that they always have a Hamiltonian cycle.
Intuitively, it can be found in a visibility embedding of a non-path \PVG\ in polynomial time by going from the outermost to the innermost convex hull of the points in the embedding (see \cite[Theorem~1]{Ghosh201517} for the details). 

 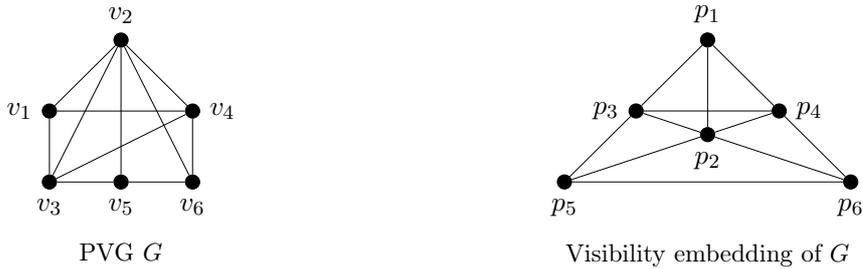
\begin{figure}[t]
     \centering
     \resizebox{\linewidth}{!}{
     \begin{tabular}{cc}
     	\begin{subfigure}{.5\linewidth}
        	\centering
     		\begin{tikzpicture}[scale=1]
                        \tikzstyle{knoten}=[circle,draw,fill,minimum size=5pt,inner sep=2pt]
                        \node[knoten,label={[label distance=0cm]180:$v_1$}] (K-1) at (0,0) {};
                        \node[knoten,label={[label distance=0cm]90:$v_2$}] (K-2) at (1,1) {};
                        \node[knoten,label={[label distance=0cm]0:$v_4$}] (K-4) at (2,0) {};
                        \node[knoten,label={[label distance=0cm]270:$v_3$}] (K-3) at (0,-1) {};
                        \node[knoten,label={[label distance=0cm]270:$v_5$}] (K-5) at (1,-1) {};
                        \node[knoten,label={[label distance=0cm]270:$v_6$}] (K-6) at (2,-1) {};
                        \foreach \i / \j in {1/3,1/4,1/2,2/3,2/4,2/5,2/6,3/4,3/5,4/6,5/6}{
                            \path (K-\i) edge[-] (K-\j);
                        }
           	\end{tikzpicture}
           	\caption{\PVG~$G$}
        	\end{subfigure}
            &
			\begin{subfigure}{.5\linewidth}
                     \centering
                     \begin{tikzpicture}[scale=1] 
                     \tikzstyle{knoten}=[circle,draw,fill,minimum size=5pt,inner sep=2pt]
                     \node[knoten,label={[label distance=0cm]90:$p_1$}] (K-1) at (2,2) {};
                     \node[knoten,label={[label distance=0cm]270:$p_2$}] (K-2) at (2,0.66666) {};
                     \node[knoten,label={[label distance=0cm]180:$p_3$}] (K-3) at (1,1) {};
                     \node[knoten,label={[label distance=0cm]0:$p_4$}] (K-4) at (3,1) {};
                     \node[knoten,label={[label distance=0cm]270:$p_5$}] (K-5) at (0,0) {};
                     \node[knoten,label={[label distance=0cm]270:$p_6$}] (K-6) at (4,0) {};
                     \foreach \i / \j in {1/3,1/4,1/2,2/3,2/4,2/5,2/6,3/4,3/5,4/6,5/6}{
                            \path (K-\i) edge[-] (K-\j);
                        }
                        \end{tikzpicture}
            \caption{Visibility embedding of $G$}
			\end{subfigure}
            \end{tabular}}
            \caption{A point visibility graph~$G$ and its visibility embedding in the plane.}
            \label{figure:PVG}
        \end{figure}

\subsection{Related Work}
Structural properties of \PVGS\ have been thoroughly researched. % 
\citet{Ghosh201517} have shown among other results that non-path \PVGS\ always have a Hamiltonian cycle. 
\citet{kara2005chromatic} gave results on the chromatic number of \PVGS. They showed that for \PVGS\ with clique size two and three, the chromatic number is two and three, respectively, and characterize those \PVGS\ that are 2- and 3-colorable. It follows that \textsc{3-Colorability} is polynomial-time solvable on \PVGS\ when a visibility embedding is given.
On the other side, \citet{DR17} showed that it is NP-hard to decide whether a given \PVG\ is $k$-colorable for $k \geq 5$.
Recently, the case of $k = 4$ was shown to be polynomial-time solvable (on a given visibility embedding)~\cite{DR17a}.
Furthermore, \citet{Pfender2008} showed that for \PVGS\ with clique size six, the chromatic number can be arbitrary large. 
 
As regards other graph problems, \citet{Ghosh201517} showed that \textsc{Vertex Cover}, \textsc{Independent Set} and \textsc{Maximum Clique} remain NP-hard on \PVGS.
Notably, the recognition problem for \PVGS\ was shown by \citet{CH17} to be complete for the existential theory of the reals~$\exists\mathbb{R}$ (which implies NP-hardness; note also that $\text{NP} \subseteq \exists\mathbb{R}\subseteq \text{PSPACE}$~\cite{CH17}).

Besides \PVGS\, the structure of other visibility graphs like line segments or polygons has also been investigated~\cite{Rourke:1987, GG13}.

\subsection{Our Contribution and Organization}
We study the complexity of several classical computational graph problems on the graph class of \PVGS. Due to the fact that \PVGS\ are Hamiltonian, problems like \textsc{Longest Path} and \textsc{Hamilton Path} are trivial. Nevertheless, many graph problems remain NP-hard on \PVGS.
In \autoref{section:nphardness} we prove NP-hardness of \textsc{Feedback Vertex Set}, \textsc{Longest Induced Path}, \textsc{Bisection}, and a restricted version of \textsc{$\mathcal{F}$-free Vertex Deletion} on \PVGS.
Herein, we build upon a general reduction idea introduced by \citet{Ghosh201517}.
In \autoref{section:dominatingset}, we briefly discuss the \textsc{Dominating Set} problem and show that it is unlikely to be NP-hard at least on a subclass of point visibility graphs.
We close in \autoref{section:outlook} by pointing to some open questions.
\section{NP-hardness Results}
\label{section:nphardness}

In this section we prove that several graph problems remain NP-hard when restricted to \PVGS.
Our hardness results follow from a transformation (mapping arbitrary graphs to \PVGS) that was first introduced by \citet{Ghosh201517} to show NP-hardness of \textsc{Vertex Cover}, \textsc{Independent Set} and \textsc{Maximum Clique} on \PVGS\ and was already used in complexity studies of other point set problems~\cite{FKNN17}. 
This transformation, henceforth called~$\Phi$, allows us to prove NP-hardness for the following problems on \PVGS: \textsc{Feedback Vertex Set}, \textsc{Longest Induced Path}, \textsc{Bisection} and \textsc{$\mathcal{F}$-free Vertex Deletion} for certain sets $\mathcal{F}$ (we are not aware of any other results concerning these problems on Hamiltonian graphs). The formal definition is as follows.

\begin{definition}[Transformation~$\Phi$]
Given a graph $G=(V,E)$, we add a vertex~$b_{uv}$ for every vertex pair~$u\neq v\in V$ with~$\{u,v\}\not\in E$ and we connect~$b_{uv}$ to all vertices in~$V$.
We will call~$b_{uv}$ a \emph{blocker} because it blocks the visibility between~$u$ and~$v$ in a corresponding visibility embedding.
In total, $\binom{|V|}{2}-|E|$ blockers are added to~$G$.
Finally, all blockers are connected to each other to obtain~$\Phi(G)$.
\end{definition}% 
A small example of the transformation $\Phi$ is shown in \ref{fig:transformationExample}. It is not hard to see that the resulting graph~$\Phi(G)$ can be computed in polynomial time. Keep in mind below that all added blockers in~$\Phi(G)$ form a clique and that the graph~$G$ is an induced subgraph of~$\Phi(G)$.
Furthermore, $\Phi(G)$ is always a \PVG~\cite{Ghosh201517, FKNN17}. A proof sketch is as follows.
Let~$v_1,\ldots,v_n$ be~$n$ distinct points in general position (e.g.\ on a circle) corresponding to the vertices of~$G$.
We can now add blockers inductively as follows:
For a non-edge~$\{u,v\}\not\in E$ consider the line segment~$uv$ defined by~$u$ and~$v$. It is clear that we can always choose a blocker~$b_{uv}$ on this segment such that~$b_{uv}$ is not lying on any other line defined by any two other points introduced so far, since there are only finitely many intersection points of these lines with~$uv$.
Using a similar argument, we can show that, for each PVG, we can add a \emph{universal} vertex~$u$, a vertex connected to all other vertices, while maintaining the PVG property. We use this observation below.

In the following, we will use the transformation $\Phi$ in polynomial-time reductions to prove NP-hardness for the above mentioned problems. It is clearly computable in polynomial time, since it only involves adding a polynomial number of vertices and edges to the input graph in a trivial way.

\begin{figure}[t]
    \centering
        \resizebox{\linewidth}{!}{
            \begin{tabular}{ccc}
            	\begin{subfigure}{.5\linewidth}
                \centering
                \begin{tikzpicture}[scale=1]
                \tikzstyle{knoten}=[circle,draw,fill,minimum size=5pt,inner sep=2pt]
                \node[knoten,label={[label distance=0cm]180:$v_1$}] (K-1) at (0,0) {};
                \node[knoten,label={[label distance=0cm]0:$v_2$}] (K-2) at (2,0) {};
                \node[knoten,label={[label distance=0cm]00:$v_3$}] (K-3) at (2,2) {};
                \node[knoten,label={[label distance=0cm]180:$v_4$}] (K-4) at (0,2) {};
                \foreach \i / \j in {1/4, 4/3, 2/3, 1/3}{
                    \path (K-\i) edge[-] (K-\j);
                }
                \end{tikzpicture}
                \caption{General graph $G$}
                \end{subfigure}
                &

                \begin{subfigure}{.5\linewidth}
                	\centering
                    \begin{tikzpicture}[scale=1]
                    \tikzstyle{knoten}=[circle,draw,fill,minimum size=5pt,inner sep=2pt]
                    \node[knoten,label={[label distance=0cm]180:$v_1$}] (K-1) at (0,0) {};
                    \node[knoten,label={[label distance=0cm]0:$v_2$}] (K-2) at (2,0) {};
                    \node[knoten,label={[label distance=0cm]00:$v_3$}] (K-3) at (2,2) {};
                    \node[knoten,label={[label distance=0cm]180:$v_4$}] (K-4) at (0,2) {};
                    \node[knoten,fill=red,label={[label distance=-.1cm]90:}] (K-5) at (1,0) {};
                    \node[knoten,fill=red,label={[label distance=-.1cm]90:}] (K-6) at (0.5,1.5) {};
                    \foreach \i / \j in {1/4, 4/3, 2/3, 1/3}{
                        \path (K-\i) edge[-] (K-\j);
                    }
                    
                    \foreach \i in {1, 2, 3, 4}{
                        \path (K-\i) edge[-] (K-5);
                        \path (K-\i) edge[-] (K-6);
                    }
                    \path (K-5) edge[-] (K-6);

                    \end{tikzpicture}
                \caption{PVG $\Phi(G)$}
                \label{figure:transformationExampleB}
                \end{subfigure}
            \end{tabular}}
        \caption{Transformation of a general graph $G$ to a PVG $\Phi(G)$. The red vertices are the blockers introduced in the transformation $\Phi(G)$. }
        \label{fig:transformationExample}
    \end{figure}
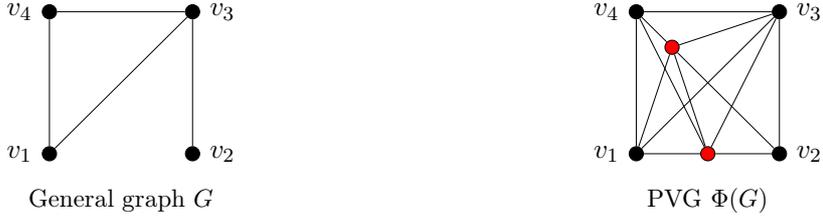

\subsection{Feedback Vertex Set}
We start with the well-known \textsc{Feedback Vertex Set} problem which is NP-hard on general graphs~\cite{GJ79}.

\problemdef{Feedback Vertex Set}
{A graph $G=(V,E)$ and an integer $k \in \mathbb{N}$.}
{Is there a subset $V' \subseteq V$ with $|V'| \leq k$ such that all cycles in~$G$ contain at least one vertex of~$V'$?}

We show that it is also NP-hard on \PVGS.

\begin{theorem}
\textsc{Feedback Vertex Set} on \PVGS\ is NP-hard.
\end{theorem}

\begin{proof}
Given a \textsc{Feedback Vertex Set} instance $(G,k)$, we construct an instance~$(G',k')$ of \textsc{Feedback Vertex Set} on \PVGS\ as follows. 
We set $G':=\Phi(G)$ and $k':=k + |B|$, where~$B$ is the set of blockers that where introduced by the transformation $\Phi$.

Let~$(G,k)$ be a yes-instance and let~$v_1,\ldots,v_k$ be vertices in~$G$ such that $G-\{v_1,\ldots,v_k\}$ does not contain a cycle. 
Then, removing the vertices $v_1, \ldots,v_k$ along with all the blockers~$B$ from $\Phi(G)$ yields the acyclic graph $G-\{v_1,\ldots,v_k\}$. Hence,~$(G',k')$ is a yes-instance.

If~$(G,k)$ is a no-instance, then we need to remove at least $k+1$ vertices to delete all cycles in~$G$. Hence, there are more than~$k+2$ vertices in~$G$ since, otherwise, deleting~$k$ arbitrary vertices removes all cycles.
Moreover, deleting any set of~$k+1$ vertices in~$G$ always leaves an edge in the remaining graph since, otherwise, we could only delete~$k$ of these vertices and still remove all cycles.

Now note that we always have to delete at least~$|B|-2$ of the blockers in~$\Phi(G)$ since 
all blockers are connected to all other vertices and any three of them form a cycle.
However, we know that we cannot delete all blockers because we need at least~$k+1$ vertex deletions in order to delete all cycles in~$G$.
If at least one blocker remains, then we can delete at most~$k+1$ vertices in~$G$. But then, there still exists an edge in~$G$ which forms a cycle with the one remaining blocker.
If two blockers remain, then we can delete at most~$k+2$ vertices in~$G$. Hence, there still remains a vertex in~$G$ forming again a cycle with the two remaining blockers.
Thus, $(G',k')$ is also a no-instance.
\end{proof}

\subsection{Longest Induced Path}
In the following we define the length of a path as the number of edges it contains.

\problemdef{Longest Induced Path}
{A graph $G$ and an integer $k \in \mathbb{N}$.}
{Is there an induced path with at least~$k$ edges in~$G$?}

This problem is NP-hard on general graphs~\cite{GJ79} and remains NP-hard on \PVGS. 

\begin{theorem}
\textsc{Longest Induced Path} on \PVGS\ is NP-hard.
\end{theorem}

\begin{proof}
  Let~$(G,k)$ be an instance of \textsc{Longest Induced Path}.
  If~$k\le 2$, then we solve the instance in polynomial time and return a trivial constant-size yes- or no-instance accordingly.
  For~$k\ge 3$, we construct the instance $(\Phi(G),k)$.

A blocker is connected to all other vertices and forms a cycle of length three with each pair of vertices that are adjacent to each other.
Thus, no blocker lies on any induced path of $\Phi(G)$ of length at least three.
That is, any induced path of~$\Phi(G)$ of length at least three is an induced path in~$G$.
Consequently, $(G,k)$ is a yes-instance if and only if~$(\Phi(G),k)$ is a yes-instance.
\end{proof}

\subsection{Bisection}
The \textsc{Bisection} problem is to partition the vertices of a graph into two equally sized disjoint subsets such that the number of edges between these two vertex subsets is minimized.
The problem is formally defined as follows.

\problemdef{Bisection}
{A graph $G=(V,E)$ and an integer $k \in \mathbb{N}$.}
{Is there a partition $(U,W)$ of $V$, that is, $U, W \subseteq V$, $U \cap W = \emptyset$, and $U \cup W = V$, such that $|U| = |W|$ and $|\{\{u,w\}\in E \mid u \in U \wedge w \in W\}| \leq k$? }

On \PVGS\ the \textsc{Bisection} problem translates to partitioning the points in the plane into
equal-size subsets such that the number of pairs of points from both subsets that can see each other is minimized.
To show NP-hardness of this problem on \PVGS, we use an idea by \citet{GAREY1976237} who showed NP-hardness of \textsc{Bisection} by reducing from \textsc{Max Cut}. 

\problemdef{Max Cut}
{A graph $G=(V,E)$ and an integer $k \in \mathbb{N}$.}
{Is there a partition $(U, W)$ of $V$ such that $|\{\{u,w\}\in E \mid u \in U \wedge w \in W\}| \geq k$?} 

\begin{theorem}
\textsc{Bisection} on \PVGS\ is NP-hard.
\end{theorem}

\begin{proof}
Given a \textsc{Max Cut} instance $(G=(V,E),k)$, we construct an instance $(G',k')$ of \textsc{Bisection} on \PVGS\ as follows. 
To obtain $G'$, we first take the complement graph $\overline G$ of $G$ and apply the transformation $\Phi$. We obtain $\Phi(\overline G)$ with~$B$ being the set of blockers introduced by~$\Phi$. 
If $|B| < |V|$, then we add~$|V|-|B|$ additional vertices to $\Phi(\overline G)$ and connect them to all other vertices.
If~$|B| \ge |V|$ and~$|B|+|V|$ is odd, then we add another universal vertex to get an even number of vertices.
As mentioned before, adding such universal vertices does not destroy the \PVG\ property. 
We obtain the \PVG\ $G'=(V',E')$ with~$V'=V\cup B'$, where~$B'$ is the set of all added vertices (original blockers and additional universal vertices).
Finally, we set $k':=(\frac{1}{2}|V'|)^2 - k$. 

If~$(G,k)$ is a yes-instance of \textsc{Max Cut}, then there exists a partition $(U,W)$ of $V$ such that $|\{\{u,w\}\in E \mid u \in U \wedge w \in W\}| \geq k$. 
In the complement graph~$\overline G$ it holds that there are at least~$k$ edges missing between the vertices in~$U$ and~$W$. 
Now we choose a subset~$U'\subseteq (U\cup B')$ with~$U\subseteq U'$ and a subset $W'\subseteq (W\cup B')$ with $W\subseteq W'$ such that $|U'|=|W'|=\frac{1}{2}|V'|$ and $U' \cup W' = V'$. 
Informally speaking, we fill up the vertex sets~$U$ and~$W$ with the vertices in~$B'$ to obtain two disjoint equally sized vertex sets~$U'$ and $W'$. This is always possible since~$|B'|\ge |V|$ by construction.
Clearly, it holds $|\{\{u,w\} \in E' \mid u \in U' \wedge w \in W'\}| \leq (\frac{1}{2}|V'|)^2-k $
and hence~$(G',k')$ is a yes-instance.

If~$(G',k')$ is a yes-instance, then there exist two vertex sets $U'$ and $W'$ with $U' \cup W' = V'$, $|U'| = |W'|$, and $\{\{u,w\}\in E'\mid u \in U'\wedge w \in W'\} \leq k' = (\frac{1}{2} |V'|)^2 - k$. 
Since the vertices in~$B'$ are universal vertices, the~$k$ missing edges between~$U'$ and~$W'$ can only be between vertices of~$\overline G$. 
Set $U := U' \cap V$ and $W :=  W'\cap V$. Clearly, we have $U \cup W= V$ and $\{\{u,w\}\in E\mid u \in U \wedge w \in W\} \geq k$. Hence,~$(G,k)$ is a yes-instance.
\end{proof}

\subsection{$\mathcal{F}$-free Vertex Deletion}
\label{sec:FfreeDeletion}
In this section, we consider a general graph problem called \textsc{$\mathcal{F}$-free Vertex Deletion}.
In the following, $\mathcal{F}$ denotes a finite set of graphs, $K_i$, $i \in \mathbb{N}$, denotes the complete graph on~$i$~vertices, and~$K_{i,j}$, $i,j\in\mathbb{N}$ denotes the complete bipartite graph with two partite sets containing~$i$ and $j$ vertices, respectively. We say that a graph~$G$ is \emph{$\mathcal{F}$-free} if no $H \in \mathcal{F}$ occurs as an induced subgraph of~$G$. The \textsc{$\mathcal{F}$-free Vertex Deletion} problem is then defined as follows:

\problemdef{$\mathcal{F}$-free Vertex Deletion}
{A graph $G=(V,E)$ and an integer $k \in \mathbb{N}$.}
{Is there a subset of vertices $X \subseteq V$ with $|X| \leq k$ such that the graph $G-X := G[V\setminus X]$ obtained by deleting all vertices in~$X$ from~$G$ is $\mathcal{F}$-free?}

This is a generic graph problem that generalizes (depending on the choice of~$\mathcal{F}$) various fundamental graph problems such as \textsc{Vertex Cover} ($\mathcal{F}=\{K_2\}$) or \textsc{Cluster Vertex Deletion} ($\mathcal{F}=\{P_3\}$, that is, $\mathcal{F}$ consists of a single path on three vertices). In a vertex-deletion problem we are given a graph~$G$ and an integer~$k$ and we want to decide whether we can delete at most~$k$ vertices from~$G$ such that the resulting graph has a certain fixed property~$\Pi$. \citet{LEWIS1980219} showed that vertex-deletion problems are NP-hard whenever $\Pi$ is a \emph{nontrivial} property, that is, there are infinitely many graphs satisfying~$\Pi$ and infinitely many graphs not satisfying~$\Pi$. Hence, \textsc{$\mathcal{F}$-free Vertex Deletion} is NP-hard on general graphs if there are infinitely many $\mathcal{F}$-free graphs and infinitely many graphs that are not $\mathcal{F}$-free.

For point visibility graphs, we prove the following theorem.
\begin{theorem}
  \label{thm:F-free}
  \textsc{$\mathcal{F}$-free Vertex Deletion} is NP-hard on \PVGS\ in each of the three individual cases where
  \begin{enumerate}[(i)]
  	\item\label{no comp} $\mathcal{F}$ contains no complete graphs,
        \item\label{Kt} $\mathcal{F}$ contains a $K_t$ with~$t\ge 3$ and no graph in $\mathcal{F}$ can be made~$K_{t-1}$-free with less than two vertex deletions, 
  	\item\label{KtK1t} $\mathcal{F}$ contains $K_t$ and $K_{1,t'}$ with $t \geq 3$, $t' \geq 2$ and no $K_2$-free graphs.
  \end{enumerate} 
\end{theorem}

\begin{proof}
  We prove the three cases separately.
\begin{compactenum}[(i)]

\item % 
Let~$\mathcal{F'}$ be the set resulting from~$\mathcal{F}$ by removing all universal vertices (that is, vertices that are adjacent to all other vertices) from every graph in~$\mathcal{F}$.
Clearly, there are infinitely many graphs that are not $\mathcal{F'}$-free (e.g.~by adding vertices to any graph in~$\mathcal{F'}$).
All complete graphs, however, are~$\mathcal{F'}$-free.
Therefore, $\mathcal{F'}$-free Vertex Deletion is NP-hard on general graphs~\cite{LEWIS1980219}.
We give a reduction from \textsc{$\mathcal{F'}$-free Vertex Deletion} on general graphs to \textsc{$\mathcal{F}$-free Vertex Deletion} on PVGs.

Given an instance~$(G,k)$, we construct the instance~$(G',k)$, where~$G'$ is the \PVG\ obtained by adding $k + \ell$ additional universal vertices to~$\Phi(G)$, where~$\ell$ is the maximum number of universal vertices of any graph in $\mathcal{F}$. Note that $\ell$ is a fixed constant (for every~$\mathcal{F}$) and, hence, we can add these vertices in polynomial time. As mentioned before, adding universal vertices to a PVG yields another PVG.

Assume that there is a size-$k$ vertex subset $X \subseteq V$ such that~$G - X$ is $\mathcal{F'}$-free.
We claim that $G' - X$ is $\mathcal{F}$-free.
If $G' - X$ contains $H\in \mathcal{F}$ as an induced subgraph,
then all vertices of that subgraph that are not universal are also contained in $G - X$.
This implies that $G - X$ contains an induced subgraph~$H' \in \mathcal{F'}$ obtained from~$H$ by removing all universal vertices, which is a contradiction.
Therefore, $G' - X$ is indeed $\mathcal{F}$-free.

Now assume that~$G'$ contains a size-$k$ vertex set~$X$ such that $G' - X$ is $\mathcal{F}$-free.
Since~$G'$ contains at least $k+\ell$ universal vertices, it follows that~$G' - X$ contains at least~$\ell$ universal vertices.
Therefore, $G' - X$ contains all the required universal vertices for every graph $H \in \mathcal{F}$.
This implies that for all $H \in \mathcal{F}$, the graph $G' - X$ does not contain a copy of~$H'\in\mathcal{F'}$, where~$H'$ is again obtained from~$H$ by removing all universal vertices.
Hence, also $G - X$ is $\mathcal{F'}$-free.

\item Note that~$\mathcal{F}$ contains no~$K_2$-free graphs. Hence, all edgeless graphs are $\mathcal{F}$-free. Moreover, every complete graph with at least~$t$ vertices is not $\mathcal{F}$-free, implying that \textsc{$\mathcal{F}$-free Vertex Deletion} is NP-hard on general graphs~\cite{LEWIS1980219}.
We reduce the problem on general graphs to its restriction on PVGs.
Given an instance~$(G,k)$, we construct the instance $(\Phi(G),k':=k+|B|)$, where $B$ is the set of blockers introduced in $\Phi(G)$.

Assume that $G$ contains a vertex subset~$X\subseteq V$ of size~$k$ such that $G-X$ is $\mathcal{F}$-free. Removing~$X$ along with all blockers~$B$ from $\Phi(G)$ clearly yields an $\mathcal{F}$-free graph.

Now, let $\Phi(G)$ contain a vertex subset~$X$ of size $k+|B|$ such that $\Phi(G)-X$ is $\mathcal{F}$-free.
If~$B\subseteq X$, then $X\setminus B$ is a set of~$k$ vertices such that~$G-(X\setminus B)$ is $\mathcal{F}$-free.
If~$B\not\subseteq X$, then it holds that~$1 \le t':=|B\setminus X| < t$ since otherwise the blockers contain a~$K_t$.
Also, it follows that~$G-(X\setminus B)$ is~$K_{t-t'}$-free since all blockers are universal vertices.
Let~$X'\subseteq (X\setminus B)$ contain arbitrary~$k$ vertices from~$X\setminus B$.
Note that~$G-X'$ is~$K_t$-free and can be made~$K_{t-t'}$-free by deleting at most~$|(X\setminus B) \setminus X'| = k+t' - k= t'$ vertices. Hence,~$G-X'$ can be made~$K_{t-1}$-free by at most one vertex deletion.
It follows that $G-X'$ cannot contain any graph in~$\mathcal{F}$.

\item Note that all edgeless graphs are $\mathcal{F}$-free and all complete graphs with at least~$t$ vertices are not. Hence, $\mathcal{F}$-free Vertex Deletion is NP-hard on general graphs~\cite{LEWIS1980219}. We reduce the general case to the restriction on PVGs.

  Let~$(G,k)$ be the input instance.
  By Ramsey's theorem, there exists a number $R(t,t')$ such that every graph with at least~$R(t,t')$ vertices contains a~$K_t$ or an edgeless induced subgraph with~$t'$ vertices~\cite{Ramsey30}.
  First, assume that $n < k + R(t,t')$.
  Then, there are $$\binom{n}{k} \le \binom{k+R(t,t')}{k} = \binom{k+R(t,t')}{R(t,t')}\in O\left(k^{R(t,t')}\right)$$ possible ways to choose a subset of~$k$ vertices.
  Since~$R(t,t')$ is a constant, we can solve the input instance in polynomial time by brute force and output a trivial yes- or no-instance.
  If $n\geq k+ R(t,t')$, then we construct the instance $(\Phi(G),k':=k+|B|)$, where $B$ is the set of blockers introduced in $\Phi(G)$.
  
First, assume that $G$ contains a size-$k$ vertex subset~$X$ such that $G-X$ is $\mathcal{F}$-free. Then, again, removing $X$ along with all blockers from $\Phi(G)$ yields an $\mathcal{F}$-free graph.

Now assume that $X$ is a vertex set of size $k+|B|$ such that $\Phi(G) - X$ is $\mathcal{F}$-free. First, we claim that $X$ contains all blockers in $\Phi(G)$.
The graph~$\Phi(G) - X$ contains at least $n+|B|- (k + |B|) = n-k \geq R(t,t')$ vertices.
Since it is~$K_t$-free, it follows that it contains~$t'$ pairwise non-adjacent vertices.
These~$t'$ vertices have to be from $G$, because blockers are universal.
If $\Phi(G) - X$ additionally contains a blocker, then there exists a~$K_{1,t'}$, which is not possible.
Therefore, $X$ contains all blockers and only~$k$ vertices from~$G$ and $G - (X \setminus B)$ is $\mathcal{F}$-free.
\end{compactenum}\end{proof}

We remark that Case~(ii) of \Cref{thm:F-free} subsumes the case that~$\mathcal{F}$ contains only complete graphs (excluding~$K_1$). % 

\section{Dominating Set on Point Visibility Graphs}
\label{section:dominatingset}
In this section, we focus on the following NP-hard problem~\cite{GJ79}:

\problemdef{Dominating Set}
{A graph $G=(V,E)$ and a parameter $k \in \mathbb{N}$}
{Is there a set $D \subseteq V$ with $|D| \leq k$ such that every vertex is contained or has at least one neighbor in~$D$?}

\noindent\textsc{Dominating Set} in PVGs can be interpreted as a guarding problem in which the vertices represent places to be observed and we want to select a small number of observation posts among them that see all other places (also known as the \emph{art gallery problem}~\cite{Ghosh07}). 

Notably, the transformation~$\Phi$ which has been used in \autoref{section:nphardness} to show NP-hardness for several graph problems does not work for \textsc{Dominating Set}:
If the input graph~$G$ is not complete, then at least one blocker will be added to~$G$.
Then, regardless of the input graph~$G$, the \PVG~$\Phi(G)$ always has a dominating set of size one containing a single blocker.
In fact, the complexity of \textsc{Dominating Set} on \PVGS\ remains unresolved.
However, for a restricted subclass of \PVGS, there exists some indication that it is unlikely to be NP-hard, which we will briefly discuss.
To this end, we define the subclass of grid point visibility graphs.

\begin{definition}
	An $n \times m$ grid point visibility graph (GPVG) is a \PVG\ that is induced by the point set $P=\{(x,y) \mid 1 \leq x \leq n\ \wedge 1 \leq y \leq m\}$. For $n = m$, we call the graph a \textit{square} GPVG.
\end{definition}

Interestingly, for the point set that induces an $n \times m$ GPVG, it is not hard to see that two points $p_1 = (x_1,y_1)$ and $p_2 = (x_2,y_2)$ are visible to each other if and only if the numbers $|x_1-x_2|$ and $|y_1-y_2|$ are relatively prime.
Two integers~$a, b \in \mathbb{N}$ are \emph{relatively prime} if there is no prime number that divides both of them.
Therefore, the structure of GPVGs can be analyzed using number theory.
For example, the following bounds on the size of a minimum dominating set are already known:

\begin{theorem}[\citet{Abbott1974199}]
\label{theorem:abbott}
Let $f(n)$ be the size of a minimum dominating set of the $n\times n$ GPVG. Then, for sufficiently large~$n$, it holds 
\begin{align*} 
\frac{\log(n)}{2\log\log(n)} < f(n) < 4 \log(n).
\end{align*}
\end{theorem}

Due to the upper bound in \autoref{theorem:abbott}, we can conclude that \textsc{Dominating Set} is unlikely to be NP-hard on square GPVGs since this implies that all problems in NP are solvable in quasi-polynomial time, that is, in time $2^{(\log n)^c}$ for some constant~$c$, where $n$ is the input size, which contradicts the \emph{Exponential Time Hypothesis}~\cite{IP01}.

\begin{corollary}
If \textsc{Dominating Set} on square GPVGs is NP-hard, then every problem in NP can be solved in quasi-polynomial time.
\end{corollary}
\begin{proof}
  Given an $n \times n$ GPVG, we know by~\autoref{theorem:abbott} that the size of an optimal dominating set is upper-bounded by~$4 \log n$.
  Hence, an optimal solution can be found by brute-force in quasi-polynomial time $O(n^{4\log n} \cdot \text{poly}(n))$. 
If \textsc{Dominating Set} is indeed NP-hard on GPVGs, then every problem in NP is quasi-polynomial-time solvable by reducing it first in polynomial time to~\textsc{Dominating Set} on a square GPVG which can then be solved by the above brute-force algorithm.
\end{proof}
\noindent Note that the encoding of the graph is crucial above: We could encode a square GPVG in a single integer but then the running time of the brute-force algorithm would not be quasi-polynomial in the input length anymore.  

Another corollary is related to fixed-parameter tractability. A problem is \emph{fixed-parameter tractable} with respect to some parameter~$k$, an integer-valued function of the input, if it admits an algorithm with running time~$f(k) \cdot n^c$, where $c$ is a constant and $n$ the input size. We obtain fixed-parameter tractability on GPVGs with respect to the size of the dominating set. Note that \textsc{Dominating Set} is W[2]-hard on general graphs and thus presumably not fixed-parameter tractable, see \citet{DF13}.
    
\begin{corollary}
\textsc{Dominating Set} on square GPVGs is fixed-parameter tractable with respect to the sought solution size.
\end{corollary}
\begin{proof}
  Given a \textsc{Dominating Set} instance $(G,k)$, where~$G$ is an $n \times n$ GPVG,
  we know that if $k < \frac{\log(n)}{2\log\log(n)}$, then we can answer ``no'' due to the lower bound in \autoref{theorem:abbott}. Otherwise, we have~$k \ge \frac{\log(n)}{2\log\log(n)}$ and consequently $n \le g(k)$ for some function~$g$.
  This implies that \textsc{Dominating Set} on square GPVGs is fixed-parameter tractable with respect to~$k$ since we can check every possible solution in a running time only depending on~$k$.
\end{proof}

\begin{figure}[!t]
  \begin{center}
	\begin{tikzpicture}[scale=.5] 
                \tikzstyle{knoten}=[circle,draw,fill,inner sep=2pt]

                \draw (-3,0) -- (4,0);
                \draw (-3,-3) -- (4,4);
                \draw (-1.5,-3) -- (2,4);
                \draw (0,-3) -- (0,4);
                
                \node[knoten] at (0,0) {};
                
                \node[knoten,color=red] at (1,1) {};
                \node[knoten]  at (2,2) {};
                \node[knoten]  at (3.5,3.5) {};
                
                \node[knoten,color=red]  at (1,2) {};
                \node[knoten]  at (1.3,2.6) {};
                \node[knoten]  at (1.7,3.4) {};
                
                \node[knoten]  at (0,1.5) {};
                \node[knoten]  at (0,3.3) {};

                \node[knoten]  at (2,0) {};
                \node[knoten]  at (3.7,0) {};

	\end{tikzpicture}
      \begin{tikzpicture}[scale=.5]
                \tikzstyle{knoten}=[circle,draw,fill,inner sep=2pt]

                \draw (-3,0) -- (4,0);
                \draw (-3,-3) -- (4,4);
                \draw (-1.5,-3) -- (2,4);
                \draw (0,-3) -- (0,4);
                \draw (-3,1.5) -- (4,-2);
                
                \node[knoten] at (0,0) {};
                
                \node[knoten,color=red] at (1,1) {};
                \node[knoten]  at (2,2) {};
                \node[knoten]  at (3.5,3.5) {};
                
                \node[knoten,color=red]  at (1,2) {};
                \node[knoten]  at (1.3,2.6) {};
                \node[knoten]  at (1.7,3.4) {};
                
                \node[knoten]  at (0,1.5) {};
                \node[knoten]  at (0,3.3) {};

                \node[knoten,color=red]  at (2,0) {};
                \node[knoten]  at (3.7,0) {};

                \node[knoten]  at (1,-0.5) {};
                \node[knoten]  at (2,-1) {};
                \node[knoten]  at (2.6,-1.3) {};

	\end{tikzpicture}
      \begin{tikzpicture}[scale=.5]
        \tikzstyle{knoten}=[circle,draw,fill,inner sep=2pt]

        \draw (-3,0) -- (4,0);
        \draw (-3,-3) -- (4,4);
        \draw (-1.5,-3) -- (2,4);
        \draw (0,-3) -- (0,4);
        \draw (-3,1.5) -- (4,-2);
        \draw (-3,3) -- (3,-3);
        
        \node[knoten] at (0,0) {};
        
        \node[knoten,color=red] at (1,1) {};
        \node[knoten]  at (2,2) {};
        \node[knoten]  at (3.5,3.5) {};
        
        \node[knoten,color=red]  at (1,2) {};
        \node[knoten]  at (1.3,2.6) {};
        \node[knoten]  at (1.7,3.4) {};
        
        \node[knoten]  at (0,1.5) {};
        \node[knoten]  at (0,3.3) {};

        \node[knoten]  at (2,0) {};
        \node[knoten]  at (3.7,0) {};

        \node[knoten,color=red]  at (1,-0.5) {};
        \node[knoten]  at (2,-1) {};
        \node[knoten]  at (2.6,-1.3) {};

        \node[knoten,color=red]  at (1.1,-1.1) {};
        \node[knoten]  at (1.5,-1.5) {};
        \node[knoten]  at (2.5,-2.5) {};

	\end{tikzpicture}
      \begin{tikzpicture}[scale=.5]
        \tikzstyle{knoten}=[circle,draw,fill,inner sep=2pt]

        \draw (-3,0) -- (4,0);
        \draw (-3,-3) -- (4,4);
        \draw (-1.5,-3) -- (2,4);
        \draw (0,-3) -- (0,4);
        \draw (-3,1.5) -- (4,-2);
        \draw (-3,3) -- (3,-3);
        \draw (-2,4) -- (1.5,-3);
        
        \node[knoten] at (0,0) {};
        
        \node[knoten,color=red] at (1,1) {};
        \node[knoten]  at (2,2) {};
        \node[knoten]  at (3.5,3.5) {};
        
        \node[knoten,color=red]  at (1,2) {};
        \node[knoten]  at (1.3,2.6) {};
        \node[knoten]  at (1.7,3.4) {};
        
        \node[knoten]  at (0,1.5) {};
        \node[knoten]  at (0,3.3) {};

        \node[knoten]  at (2,0) {};
        \node[knoten]  at (3.7,0) {};

        \node[knoten]  at (1,-0.5) {};
        \node[knoten]  at (2,-1) {};
        \node[knoten]  at (2.6,-1.3) {};

        \node[knoten,color=red]  at (1.1,-1.1) {};
        \node[knoten]  at (1.5,-1.5) {};
        \node[knoten]  at (2.5,-2.5) {};

        \node[knoten,color=red]  at (0.8,-1.6) {};
        \node[knoten]  at (1.2,-2.4) {};

	\end{tikzpicture}	
  \end{center}
\caption{Dominating sets (red points) within the neighborhood of a minimum degree vertex (center point) for degrees~4, 5, 6, and 7.}
\label{fig:dscsps}
\end{figure}
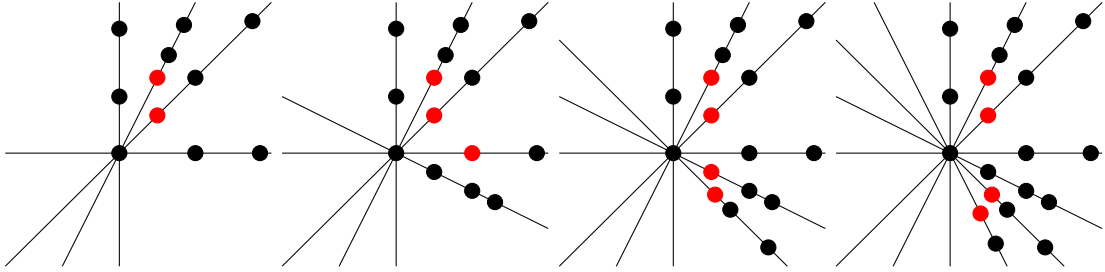

We show that there also is an upper bound for the \textsc{Dominating Set} problem on general \PVGS\ depending on the minimum vertex degree in the \PVG. In the following, we say a PVG is a \emph{non-path} PVG if it is not a path. Our result is based on the observation
that the neighborhood of any vertex in a non-path \PVG\ is a dominating set since every non-path \PVG\ has diameter two.
We can even tighten up this observation by taking a closer look at the visibility embedding of a non-path \PVG.

\begin{theorem}
  \label{thm:mindegree_bound}
For every non-path \PVG\ with minimum vertex degree~$\delta$, there exists a dominating set of size $\lfloor\frac{\delta}{2}\rfloor + 1$. 
\end{theorem}
\begin{proof}% 
  Let~$G$ be a non-path \PVG\ with minimum vertex degree~$\delta\ge 2$ and let~$v$ be a vertex of degree~$\delta$. Let~$u_1,\ldots,u_\delta$ be the neighbors of~$v$. 
  Consider a visibility embedding of~$G$ with the point set~$P$, where~$v$ corresponds to the point~$p\in P$ and~$u_1,\ldots,u_\delta$ correspond to the points~$p_1,\ldots,p_\delta$.
  Then, all points in~$P$ lie on one of the at most~$\delta$ lines defined by the pair~$(p,p_i)$ for each~$i=1,\ldots,\delta$. See \autoref{fig:dscsps} for examples.
 Let~$L_i$ denote the line defined by~$p$ and~$p_i$ and assume that~$L_1,\ldots,L_\delta$ are in clockwise order. Then, each point on~$L_i$ is visible to all the points on~$L_{i-1}$ and~$L_{i+1}$ since otherwise there would be another neighbor of~$v$ corresponding to a point (and defining a line) in between~$L_{i-1}$ and~$L_i$ or~$L_i$ and~$L_{i+1}$.
  Hence, picking the vertices~$u_i$ and~$u_{i+1}$ into the dominating set dominates all vertices corresponding to points on the four lines~$L_{i-1},\ldots,L_{i+2}$.

  Now, if~$\delta \bmod 4 = 0$, then we can select the vertices~$u_{4i-1}$ and~$u_{4i-2}$ for each~$i=1,\ldots,\lfloor\frac{\delta}{4}\rfloor$. This yields a dominating set of size~$\frac{\delta}{2}$.
  If~$\delta \bmod 4 = 1$, then we additionally select the vertex~$u_\delta$ to obtain a dominating set of size~$\frac{\delta-1}{2}+1=\lfloor\frac{\delta}{2}\rfloor+1$.
  Finally, for the case~$\delta \bmod 4 \in\{2,3\}$, we additionally select the vertices~$u_{\delta-1}$ and~$u_\delta$ and obtain a dominating set of size~$\lfloor\frac{\delta}{2}\rfloor+1$.
\end{proof}
Note that~\autoref{thm:mindegree_bound} implies that \textsc{Dominating Set} is polynomial-time solvable on \PVGS\ with constant minimum vertex degree.% 

We close this section by mentioning a conjecture by~\citet{Rourke:1987} stating that a logarithmic upper bound for the size of an optimal dominating set also holds for arbitrary \PVGS.
\begin{conjecture}[\cite{Rourke:1987}]
Every \PVG\ $G=(V,E)$ has a dominating set of size~$O(\log{|V|})$.
\end{conjecture}
\section{Outlook}
\label{section:outlook}

In this paper we made some effort towards examining point visibility graphs from the viewpoint of algorithmic complexity.
We surveyed some structural properties of point visibility graphs and started to investigate the complexity of graph problems when restricted to such graphs.
Even though some problems turn out to be efficiently (sometimes even trivially) solvable in special cases, we showed that also several classical graph problems still remain NP-hard on point visibility graphs.
Our main goal was to initiate productive research on solving computational problems for this natural graph class.
Thus, we conclude with open questions and further directions for research.

The computational complexity of several graph problems when restricted to PVGs is yet to be determined.
Among them are for example \textsc{Dominating Set} (see \autoref{section:dominatingset}), \textsc{Max Cut} or $\mathcal{F}$\textsc{-free Vertex Deletion} (see \autoref{sec:FfreeDeletion}).
An interesting open problem is the case where~$\mathcal{F}$ contains a~$K_3$ and the graph consisting of a single edge and a single isolated vertex (that is, $K_2+K_1$).
A $\{K_3, K_2+K_1\}$-free graph is either~$K_2$-free or a complete bipartite graph.
It is open whether this problem is polynomial-time solvable on \PVGS.
Also, there are open cases when~$\mathcal{F}$ contains infinitely many graphs.
Note that \textsc{Feedback Vertex Set} corresponds to the infinite set~$\mathcal{F}$ containing all cycles. For infinitely many complete graphs, NP-hardness still holds by \autoref{thm:F-free}).

Furthermore, for those problems that are NP-hard on point visibility graphs, the existence of efficient approximation algorithms certainly is an interesting question. Another interesting line of research would be the parameterized complexity of graph problems restricted to PVGs.
More specifically, for those problems that are $W[1]$- or $W[2]$-hard on general graphs,
are they fixed-parameter tractable on PVGs?
Recall that we have seen that \textsc{Dominating Set} becomes fixed-parameter tractable with respect to the solution size on quadratic grid point visibility graphs.
As a final remark, we mention that the NP-hardness reduction by~\cite{Ghosh201517} for \textsc{Independent Set} on PVGs also proves that it is $W[1]$-hard with respect to the solution size.
However, for the \textsc{Clique} problem, the parameterized complexity is still open.

\paragraph{Acknowledgments}
This work is a result of the course ``Algorithmic Research in Teams'' held at TU Berlin during the summer term 2016.

\section*{Bibliography}
\bibliographystyle{abbrvnat}
\bibliography{sources}

\end{document}